\documentclass{ifacconf}

\usepackage{lipsum}
\usepackage{amsmath,amssymb,amsfonts}
\usepackage{algorithm}
\usepackage{algorithmic}
\usepackage{textcomp}

\usepackage{graphicx}
\let\theoremstyle\relax
\usepackage{amsthm}
\usepackage{tabu}
\usepackage{verbatim}
\usepackage{esvect}
\usepackage{footnote}
\usepackage{siunitx}
\usepackage{color}
\usepackage{url}
\usepackage[makeroom]{cancel}
\usepackage[numbers]{natbib}
\usepackage{chngcntr}
\usepackage{enumerate}

\newtheorem{theorem}{Theorem}[section]
\newtheorem{corollary}[theorem]{Corollary}
\newtheorem{definition}{Definition}[section]

\newtheorem{proposition}[theorem]{Proposition}
\newtheorem{remark}{Remark}
\newtheorem*{problem statement}{Problem Statement}

\newtheorem{assumption}{Assumption}

\newcommand{\real}{\mathbb{R}}

\newcommand{\rank}{\mathrm{rank}}

\newcommand{\Cc}{{\mathcal{C}}}
\newcommand{\Nc}{{\mathcal{N}}}

\newcommand{\Xc}{{\mathcal{X}}}

\newcommand{\Ic}{{\mathcal{I}}}
\newcommand{\Pc}{{\mathcal{P}}}
\newcommand{\Jc}{{\mathcal{J}}}
\newcommand{\Rc}{{\mathcal{R}}}

\newcommand{\Kc}{{\mathcal{K}}}
\newcommand{\Sc}{{\mathcal{S}}}

\newcommand{\bK}{{\mathbf{K}}}

\newcommand{\bx}{{\mathbf{x}}}

\newcommand{\by}{{\mathbf{y}}}
\newcommand{\ba}{{\mathbf{a}}}

\newcommand{\be}{{\mathbf{e}}}
\newcommand{\bg}{{\mathbf{g}}}

\newcommand{\bu}{{\mathbf{u}}}

\newcommand{\bv}{{\mathbf{v}}}

\newcommand{\bk}{{\mathbf{k}}}
\newcommand{\bb}{{\mathbf{b}}}
\newcommand{\bff}{{\mathbf{f}}}
\newcommand{\bA}{{\mathbf{A}}}
\newcommand{\bB}{{\mathbf{B}}}

\newcommand{\bG}{{\mathbf{G}}}
\newcommand{\bH}{{\mathbf{H}}}

\newcommand{\bD}{{\mathbf{D}}}
\newcommand{\bC}{{\mathbf{C}}}
\newcommand{\bR}{{\mathbf{R}}}
\newcommand{\bQ}{{\mathbf{Q}}}

\newcommand{\mb}[1]{\mathbf{ #1 }}

\newcommand{\blambda}{{\boldsymbol{\lambda}}}
\newcommand{\bbgamma}{{\boldsymbol{\gamma}}}
\newcommand{\bbkappa}{{\boldsymbol{\kappa}}}

\newcommand\xqed[1]{%
  \leavevmode\unskip\penalty9999 \hbox{}\nobreak\hfill
  \quad\hbox{#1}}
\newcommand\demo{\xqed{$\bullet$}}

\DeclareMathOperator*{\argmin}{arg\,min}

\newcommand{\map}[3]{#1:#2 \rightarrow #3}
\newcommand{\longthmtitle}[1]{\mbox{}\emph{(#1):}}
\newcommand{\setdef}[2]{\{#1 : #2\}}

\newcommand{\norm}[1]{\left\lVert#1\right\rVert}

\allowdisplaybreaks

\begin{document}

\makeatletter
\renewcommand{\thealgorithm}{\arabic{algorithm}}
\setcounter{algorithm}{0} 
\makeatother

\begin{frontmatter}

\title{Explicit Control Barrier Function-based Safety Filters and their Resource-Aware Computation} 

\thanks[footnoteinfo]{This research was supported in part by the Boeing Strategic University Initiative and the Technology Innovation Institute (TII).
$^\dagger$ Equal contribution.}

\author[First]{Pol Mestres$^{\dagger}$}
\author[First]{Shima Sadat Mousavi$^{\dagger}$}
\author[First]{Pio Ong}
\author[First]{Lizhi Yang}
\author[First]{Ersin Da\c{s}}
\author[First]{Joel W. Burdick}
\author[First]{Aaron D. Ames}

\address[First]{California Institute of Technology, 
   Pasadena, CA 91125 USA (e-mail: \{mestres, smousavi, pioong, lzyang, ersindas, jburdick, ames\}@caltech.edu).}

\begin{abstract}
This paper studies the efficient implementation of \textit{safety filters} that are designed using control barrier functions (CBFs), which minimally modify a nominal controller to render it safe with respect to a prescribed set of states.
Although CBF-based safety filters are often implemented by solving a quadratic program (QP) in real time, the use of off-the-shelf solvers for such optimization problems poses a challenge in applications where control actions need to be computed efficiently at very high frequencies.
In this paper, we introduce a closed-form expression for controllers obtained through CBF-based safety filters.
This expression is obtained by partitioning the state-space into different regions, with a different closed-form solution in each region.
We leverage this formula to introduce a resource-aware implementation of CBF-based safety filters that detects changes in the partition region and uses the closed-form expression between changes.
We showcase the applicability of our approach in examples ranging from aerospace control to safe reinforcement learning.
\end{abstract}


\end{frontmatter}



\section{Introduction}


Most engineering systems—from humanoid robots to aerospace vehicles—must satisfy strict safety-critical constraints during operation. In control-theoretic analysis and design, these constraints are typically enforced by ensuring that the system state remains within a prescribed safe region of the state space. Several methods exist to achieve this, including control barrier functions (CBFs)~\citep{ADA-XX-JWG-PT:17}, model predictive control~\citep{JBR-DQM-MMD:17}, Hamilton–Jacobi reachability~\citep{SB-MC-SH-CJT:17}, and reference governors~\citep{EG-SDC-IK:17}.

A popular approach to designing safe controllers is through a mechanism known as \textit{safety filters}~\citep{ADA-SC-ME-GN-KS-PT:19}, whereby a nominal (potentially unsafe) controller is minimally modified in order to be safe.
A common way to implement safety filters is through CBFs.
In this case, the resulting controller is obtained at every point in the state space as the solution of an optimization problem, which is a quadratic program (QP) if the dynamics are control-affine~\citep{ADA-SC-ME-GN-KS-PT:19}.

CBF-based controllers obtained by solving QPs (often referred to as CBF-QPs) in real time have been successfully implemented in a wide range of applications, such as adaptive cruise control~\citep{ADA-JWG-PT:14}, bipedal robotic walking~\citep{SCH-XX-ADA:15},  quadrotors maneuvering~\citep{LW-ADA-ME:17}, among others, 
in part thanks to the availability of very efficient algorithms for solving QPs, such as operator splitting QP (OSQP)~\citep{BS-GB-PG-AB-SB:20}.

However, there are application domains where computing the controller by solving an optimization problem in real time is undesirable  or even  unacceptable.
For example, in aerospace control applications, 
the introduction of a QP solver in the flight control loop requires its certification in all possible flight scenarios, which can be a challenging process.
In some instances, the solver may fail to produce a solution within the required update period, or may be prone to numerical errors, which can compromise the safety and reliability of the overall system.

Another application domain where using QP solvers to implement CBF-based controllers is undesirable is reinforcement learning (RL).
In order to guarantee safety of an RL agent during the training process, it is common to add a safety filter to the RL-based policy~\citep{FPB-LB-APS:25,MA-RB-RE-BK-SN-UT:18,RC-GO-RMM-JWB:19,LY-BW-MdS-ADA:25}.
However, because RL training runs in parallel across many environments and relies on large numbers of trajectories, solving a QP at every action step becomes computationally prohibitive, especially when using GPUs.

In domains where fast computation is critical, it is desirable to implement the safety filter without relying on an online optimization solver, ideally through an explicit expression. Prior work has shown that this is possible in several important structured settings. Closed-form solutions are available for safety filters with one or two CBF constraints~\citep{XX-PT-JWG-ADA:15,XT-DVD:24,PM-YC-EDA-JC:25-jnls}. Explicit formulas are also available for some output-constraint formulations of known relative degree when the cost function is chosen appropriately~\citep{MHC-EL-ADA:25}; more recently, this line of work has been extended to broader classes of linear systems with affine safety constraints and accompanying explicit feasibility characterizations~\citep{mousavi2026structure}. Related set-wise compatibility certificates over convex regions are developed in~\citep{mousavi2025vertices}, and closed-form controllers for certain barrier-certified systems with vector-valued safety constraints are derived in~\citep{marvi2024robust}. Together, these results illustrate both the promise and the current scope of explicit safety-filter implementations: they can substantially reduce online computation, but the available formulas generally depend on specific structural assumptions. Moreover, although multiple CBFs can be merged into one~\citep{LW-ADA-ME:16,LL-DVD:19,TGM-ADA:23}, the resulting constraint may conservatively shrink the safe set and does not in general guarantee feasibility of the associated QP.

We are motivated by the desire to avoid solving the CBF-QP at every time step. Our approach is closely aligned with the broader theme of resource-aware control. In particular, we take inspiration from the event-triggered control (ETC) literature~\citep{WPMHH-KHJ-PT:12}. Since the seminal work~\citep{PT:07}, ETC has focused primarily on reducing communication~\citep{XW-ML:11} or actuation updates by transmitting control signals only when certain triggering conditions are met. 
To the best of our knowledge, there has been no ETC framework aimed at reducing computational burden. Periodic ETC (PETC)~\citep{WPMHH-MCFD-ART:08} is the closest in spirit, as it acknowledges that triggering conditions must be checked at discrete time instants due to limits in computation. However, PETC is not designed to alleviate the computational load of implementing controllers. In this paper, we introduce a PETC-inspired mechanism that alleviates the need to solve a QP at every state and utilizes computational resources only when strictly necessary, while still preserving correctness of the implemented controller.

The contributions of this paper are as follows.
First, we present a closed-form expression for controllers which are computed as the solution of QPs with an arbitrary number of CBF constraints. This formula shows that any domain where the QP is feasible can be partitioned into different regions, defined by the set of constraint indices that are active at the optimizer, and that within each region the optimal solution can be expressed algebraically.
Under mild regularity assumptions, we further show that the CBF-QP is smooth within each region and locally Lipschitz on the entire domain.
For the special case of linear systems with affine CBF constraints, we prove that the controller is affine on each region and provide a global Lipschitz constant.
Next, we observe that evaluating the closed-form CBF-QP at a given state requires identifying the region that contains that state. To address this point, we derive from the KKT conditions a set of simple tests for verifying region membership. These tests allow us to detect when the state trajectory leaves a region and thus determine when it is necessary to recompute the active set, i.e., calling on  an optimization solver. Building on this insight, we 
propose a practical, resource-aware implementation of CBF-QPs. Finally, we formalize the correctness of this resource-aware implementation and validate it extensively in a variety of case studies, ranging from aerospace control to safe reinforcement learning.

\section{Background}

In this section, we revisit CBFs and introduce the main problem we seek to solve in this paper.

Consider a control-affine system\footnote{Notation: throughout the paper, for any integer $p\in\mathbb{N}$, $[p] = \{ 1,\hdots, p \}$, 
$\alpha:(-b, a)\to\real$, with $a, b > 0$ is of extended class $\Kc$ if it is continuous, strictly increasing, and $\alpha(0) = 0$.
For a continuously differentiable function $h:\real^n\to\real$ and a vector field $\bff:\real^n\to\real^n$, we define the Lie derivative $L_{\bff}h(\bx) = \nabla h(\bx)^\top \bff(\bx)$. Given a positive definite matrix $\bR \succ 0$ and a vector $\bx$, $\norm{\bx}_{\bR} = \sqrt{\bx^\top \bR \bx}$.
For a set $\Sc$, $\text{Int}(\Sc)$ denotes its interior, and $\Pc(\Sc)$ its power set.
Given $a_1, \hdots, a_n\in\real$, $\text{diag}(a_1, \hdots, a_n)$ denotes the diagonal matrix with diagonal elements equal to $a_1, \hdots, a_n$.
}: 
\smallskip
\begin{align}\label{eq:control-affine-system}
    \dot{\bx} = \bff(\bx) + \bg(\bx) \bu,
\end{align}
\smallskip
where $\bx\in\real^n$ is the state, $\bu\in\real^m$ is the control input. Here, the drift dynamics $\map{\bff}{\real^n}{\real^n}$ and the input matrix $\map{\bg}{\real^n}{\real^{n\times m}}$ are assumed locally Lipschitz.
Given a closed set $\Cc$ describing the safe states for system~\eqref{eq:control-affine-system}, our goal is to render it forward invariant, ensuring that the state trajectory $t\mapsto\bx(t)$ remains inside the set. To this end, we recall the notion of control barrier function (CBF)~\citep{ADA-SC-ME-GN-KS-PT:19}.
\smallskip

\begin{definition}\longthmtitle{Control Barrier Function}\label{def:cbf}
    Let $h:\real^n\to\real$ be a continuously differentiable function defining a set $\Cc = \setdef{\bx\in\real^n}{h(\bx) \geq 0}$.
    The function $h$ is a CBF for~\eqref{eq:control-affine-system} if there exists an extended class $\Kc$ function $\alpha$ such that, for each $\bx\in\Cc$, there exists $\bu\in\real^m$ satisfying:
    \smallskip
    \begin{align}\label{eq:cbf-ineq}
        L_{\bff} h(\bx) + L_{\bg} h(\bx)\bu + \alpha(h(\bx)) \geq 0.
    \end{align}
\end{definition}

CBFs provide a simple framework for designing  controllers that ensure safety, i.e., forward invariance of $\Cc$.  In particular, any locally Lipschitz controller satisfying~\eqref{eq:cbf-ineq} renders $\Cc$ forward invariant~\cite[Theorem 2]{ADA-SC-ME-GN-KS-PT:19}. 

One advantage of CBFs is their flexibility to account for multiple safety constraints, each given by a different continuously differentiable function $h_i:\real^n\to\real$, for $i\in[p]$ and a corresponding zero-superlevel set $\Cc_i$.
In this case, we seek to render the set $\Cc=\bigcap_{i=1}^p \Cc_i$ safe.
Indeed, if we can find a locally Lipschitz controller that satisfies the constraints $L_{\bff} h_i(\bx) + L_{\bg} h_i(\bx)\bu + \alpha(h_i(\bx)) \geq 0$ for all $i\in[p]$, then such controller renders the set $\bigcap_{i=1}^p \Cc_i$ forward invariant.
In particular, given a desired nominal controller $\bk:\real^n\to\real^m$,
and a positive definite matrix $\bR \succ 0$,
the following controller defined as the solution of an optimization problem:
\begin{align}\label{eq:multiple-cbf-based-controller}
    \notag
    \bu^*(\bx) &= \argmin\limits_{\bu\in\real^m} \frac{1}{2}\norm{\bu-\bk(\bx)}_{\bR}^2, \\
    &L_{\bff} h_i(\bx) + L_{\bg} h_i(\bx)\bu \geq -\alpha(h_i(\bx)), \quad i\in[p],
\end{align}
\smallskip
is locally Lipschitz under mild assumptions, cf.~\citep{BJM-MJP-ADA:15,PM-AA-JC:25-ejc}.
Note that~\eqref{eq:multiple-cbf-based-controller} is a QP, often referred to as the CBF-QP or as a \textit{safety filter}.

As noted in the introduction, there are various applications where solving the QP~\eqref{eq:multiple-cbf-based-controller} in real time is undesirable. This paper seeks to develop a resource-aware implementation of the controller~\eqref{eq:multiple-cbf-based-controller} that avoids or reduces the need to execute an optimization algorithm at every time step.




\section{Explicit CBF-Based Quadratic Programs}\label{sec:explicit-cbf-qp}

In this section we provide a closed-form expression for $\bu^*$ in~\eqref{eq:multiple-cbf-based-controller}
and establish  its key properties when the constraints arise from the CBF inequalities~\eqref{eq:cbf-ineq}. While closed-form solutions are known for one or two CBF constraints~\citep{XX-PT-JWG-ADA:15,XT-DVD:24}, we extend these results to an arbitrary number of constraints. This characterization forms the basis for determining when we can avoid running an optimization algorithm to obtain $\bu^*$.

Our extension aligns with classical  Explicit MPC results~\citep{AB-MM-VD-EP:02}, which show that when a parametric QP similar to~\eqref{eq:multiple-cbf-based-controller} is feasible on some given set $\Xc\subset\real^n$, the set $\Xc$ can be partitioned into regions on which the optimizer is expressed algebraically in closed-form. Our results specialize this  to CBF-QPs and provide an explicit representation tailored to barrier-function constraints.

Before presenting our main results, we begin by introducing some notation we use, for compactness of the presentation. We define $\bb_i(\bx) = -L_{\bg}h_i(\bx)^\top$ and $a_i(\bx) = -L_{\bff}h_i(\bx) - \alpha(h_i(\bx))$ so that our CBF inequalities are simply $\bb_i(\bx)^\top\bu+ a_i(\bx) \leq 0$. With this notation, we introduce the concept of active constraints. That is, the index set of active constraints $\map{\Ic^*}{\real^n}{\Pc([p])}$ is defined as:
\begin{align*}
    \Ic^*(\bx) = \setdef{i\in[p]}{ \bb_i(\bx)^\top \bu^*(\bx) + a_i(\bx) = 0}.
\end{align*}
With this notation in place, we now make the following well-posedness assumption for our optimization problem.
\smallskip

\begin{assumption}\label{ass:assum1}
For the parametric optimization problem in~\eqref{eq:multiple-cbf-based-controller},
let $\Xc\subset\real^n$ be the domain of interest and assume:
\begin{enumerate}
    \item[(A1)] (Strict convexity) $\bR\succ0$,
    \item[(A2)] (Regularity) $\bk(\cdot)$, $L_{\bff}h_i(\cdot)$, $L_{\bg}h_i(\cdot)$, and $\alpha(h_i(\cdot))$ are $C^1$ on $\Xc$ for all $i\in[p]$.
    \item[(A3)](Linear Independence Constraint Qualification (LICQ)) For every $\bx\in \Xc$, the QP~\eqref{eq:multiple-cbf-based-controller} has an optimal solution $\bu^*(\bx)$, and the set of active constraints $\{\bb_i(\bx)\}_{i\in\Ic^*(\bx)}$ is linearly independent. 
\end{enumerate}  
\end{assumption}

Note that  (A3) implicitly requires the existence of an optimizer, and hence implies the feasibility of the QP on $\Xc$. We particularly adopt LICQ for this because it is known to ensure local Lipschitzness of the controller $\bu^*$, see \citep{PM-AA-JC:25-ejc}.
Our main result is the characterization of regions on which the constraints corresponding to a given index set $\Ic\subseteq[p]$ are active, namely:
\begin{align*}
    \Rc_{\Ic} = \setdef{\bx\in\Xc}{\Ic^*(\bx) = \Ic}.
\end{align*}
This definition is not directly useful because it requires knowing the active index set $\Ic^*(\bx)$. However, if the active constraints at a point $\bx$ are known, they can be treated as equalities, and the corresponding $\bu^*(\bx)$ is readily obtained. Unlike classical explicit MPC—where constraints are affine in both the decision variable and the parameter—CBF constraints are affine in the control but generally nonlinear in the state. As a result, the regions
$ \Rc_{\Ic}$
 need not be polyhedral in our general setup.

To formalize this, for any index set $\Ic = \{ i_1 \hdots i_k \} \subset [p]$, we define a useful shorthand notation:
\begin{align*}
    \bB_{\Ic}(\bx) \! := \! [\bb_{i_1}(\bx), \hdots, \bb_{i_k}(\bx)], \ \ba_{\Ic}(\bx) \! := \! [a_{i_1}(\bx), \hdots, a_{i_k}(\bx)]^\top.
\end{align*}
We note that whenever $\bB_{\Ic}(\bx)$ has full column rank, the matrix $\bB_{\Ic}(\bx)^\top \bR^{-1} \bB_{\Ic}(\bx)$ is invertible. In this case, the candidate optimal Lagrange multipliers and control input  associated with the index set $\Ic$ are given by:
\begin{equation}
\begin{aligned}\label{eq:lambda}
\blambda_{\Ic}(\bx)&:=\big(\bB_{\Ic}(\bx)^\top \,\bR^{-1}\bB_{\Ic}(\bx) \big)^{-1}
\big(\bB_{\Ic}(\bx)^\top \,\bk(\bx)+\ba_{\Ic}(\bx)\big),\\
\bu_{\Ic}(\bx)&:=\bk(\bx)-\bR^{-1}\bB_{\Ic}(\bx) \,\blambda_{\Ic}(\bx).
\end{aligned}
\end{equation}
The control $\bu_{\Ic}(\bx)$ represents the
optimal control input when the index set $\Ic$ is indeed $\Ic^*(\bx)$.

The following result provides an equivalent mathematical characterization of $\Rc_{\Ic}$ that does not require explicitly determining the active index set and establishes precisely when the candidate $\bu_\Ic(\bx)$ is the true optimizer $\bu^*(\bx)$.

\begin{theorem}\longthmtitle{Piecewise solution of parametric QP}
\label{thm:piecewise-solution-parametric-qp}
    Consider the parametric QP~\eqref{eq:multiple-cbf-based-controller} under Assumption~\ref{ass:assum1}.
    Then:
    \begin{enumerate}[(i)]
        \item\label{it:active-set-region} the active–set region $\Rc_{\Ic}$ is equivalently given by:
        \begin{align}\label{eq:R-Ic-expression}
        \notag
        \mathcal R_{\Ic}
        &:=\Big\{\bx\in\Xc:\ \rank(\bB_{\Ic}(\bx))=|\Ic|, 
        \blambda_{\Ic}(\bx)\ge 0, \\
        &\qquad \quad \bb_j(\bx)^\top \bu_{\Ic}(\bx)+a_j(\bx) < 0\ \ \forall j\notin\Ic \Big\}.\
        \end{align}
        \item\label{it:ustar-equals-uI} For every $\bx\in\mathcal R_{\Ic}$, one has $\bu^*(\bx)=\bu_{\Ic}(\bx)$.
\item\label{it:disjoint-intersection} For all $\Ic, \Jc \in \Pc([p])$ with $\Ic\neq\Jc$,  $\Rc_{\Ic}\cap\Rc_{\Jc} = \emptyset$.
        \item\label{it:union-of-regions} It holds that
        \(
        \Xc = \bigcup_{\Ic\in \Pc([p])} \mathcal R_{\Ic}.
        \)
    \end{enumerate}
\end{theorem}

\begin{proof}
Fix $\bx\in\Xc$.  Given the strict convexity  of~\eqref{eq:multiple-cbf-based-controller} and LICQ assumptions (A1 and A3), $\bu^*\in \real^n$ is an optimal solution if and only if there exists a corresponding $\blambda^*\in\real^p$ satisfying the KKT conditions for  the QP \eqref{eq:multiple-cbf-based-controller}:
\begin{align}
&\text{(stationarity)}\quad \bR (\bu^*-\bk(\bx))+\bB_{[p]}(\bx)\blambda^*=0, \label{eq:KKT-stationarity}\\
&\text{(primal feasibility)}\quad \bb_i(\bx)^\top\bu^*+a_i(\bx)\le 0,\ \ i\in[p], \label{eq:KKTprimal}\\
&\text{(dual feasibility)}\quad \lambda_i^*\ge 0,\ \ i\in[p], \label{eq:KKTdual}\\
&\text{(complementarity)}\quad \lambda_i^*\big(\bb_i(\bx)^\top\bu^*+a_i(\bx)\big)=0,\ \ i\in[p]. \label{eq:KKTcompslackness}
\end{align}

Assume that $\Ic^*(\bx)=\Ic$. Then  $\bB_\Ic(\bx)$ must be full rank from the LICQ assumption (A3). In addition, from~\eqref{eq:KKTcompslackness}, we have $\blambda^*_i=0$ for all $i\not\in\Ic$ because their corresponding constraints are inactive. At the same time, for active constraints where $i\in\Ic$, the primal feasibility condition~\eqref{eq:KKTprimal} holds with equality. As such, together with~\eqref{eq:KKT-stationarity}, we solve the Lagrange multiplier $\blambda^*_i$  for indices $i\in\Ic$ to be precisely the components of $\blambda_{\Ic}(\bx)$, and the optimal control $\bu^*(\bx)$ to be precisely $\bu_\Ic(\bx)$. Since $\bu_\Ic(\bx)$ is an optimal solution,  $\Ic = \Ic^*(\bx)$ is the set of active constraints, so the inactive constraints $\bb_j(\bx)^\top \bu_{\Ic}(\bx)+a_j(\bx) < 0, \forall j\notin\Ic$, must be strict inequalities.

On the other hand, assume that $\bx\in\Xc$ satisfies $\rank(\bB_{\Ic}(\bx))=|\Ic|$, $\blambda_{\Ic}(\bx) \geq 0$, and $\bb_j(\bx)^\top \bu_{\Ic}(\bx) + a_j(\bx) < 0$ for all $j\notin\Ic$.
We propose a candidate Lagrange multiplier $\blambda^*\in\real^p$ with the following components:
\begin{align*}
\lambda_{i_j}^*:=
\big(\blambda_{\Ic}(\bx)\big)_j,~i_j\in\Ic \quad \textup{and} \quad
\lambda_{i}^* := 0, i\notin\Ic.
\end{align*}
This choice meets the dual feasibility condition as we have assumed $\blambda_\Ic(\bx)\geq 0$. We can then verify the remaining KKT conditions to conclude that $\bu_\Ic(\bx)$ is the unique (by (A1)) optimal solution, that is, $\bu^*(\bx)=\bu_\Ic(\bx)$. In addition, since $\bu_\Ic(\bx)$ solves primal feasibility with equality for $i\in\Ic$ and $\bb_j(\bx)^\top \bu_{\Ic}(\bx)+a_j(\bx) < 0\ \ \forall j\notin\Ic$ by assumption, the active index set is $\Ic^*(\bx)=\Ic$ by definition.

Thus, we have proven~(\ref{it:active-set-region}) and~(\ref{it:ustar-equals-uI}). For~(\ref{it:disjoint-intersection}), since the index set $\Ic^*(\bx)$ is uniquely defined for each $\bx$, it follows that $\Rc_{\Ic}\cap\Rc_{\Jc} = \emptyset$ for any $\Ic, \Jc \subseteq [p]$. Finally for~(\ref{it:union-of-regions}), the LICQ assumption (A3) implicitly implies that an optimal solution exists at each $\bx\in\Xc$, so $\Ic^*(\bx)$ is well-defined on $\Xc$. That is, for each $\bx\in\Xc$, there exists a $\Ic\subseteq[p]$ such that $\bx\in \Rc_\Ic$, concluding the proof.
\end{proof}

Theorem~\ref{thm:piecewise-solution-parametric-qp} shows that within each region $\Rc_{\Ic}$, the controller $\bu^*$ admits a closed-form expression given in~\eqref{eq:lambda}. In addition, the theorem characterizes each region $\Rc_{\Ic}$ without requiring the explicit computation of the active set of $\Ic^*(\bx)$. In principle, one could determine the region containing a given state $\bx$ by checking all possible $\Rc_\Ic$, but the number of such sets grows combinatorially with the number of constraints $p$. 
However, the characterization of $\Rc_{\Ic}$ in the theorem provides a mechanism to verify membership to a specific region and will be 
leveraged in our resource-aware implementation later in Section~\ref{sec:resource-aware-implementation}.
\smallskip

\begin{remark}
\label{rem:general-parametric-qps}
    The results in Theorem~\ref{thm:piecewise-solution-parametric-qp} apply to several variations beyond the standard CBF-based QP~\eqref{eq:multiple-cbf-based-controller}. These include, but are not limited to, designs that use:
    \begin{enumerate}
        \item control Lyapunov function-based constraints~\citep{EDS:98} or affine input constraints in order to enforce stability or input saturation bounds, respectively;
        \item high-order or exponential CBFs (when $h$ has relative degree higher than one)~\citep{WX-CB:22,QN-KS:16};
        \item robust and adaptive CBF conditions that account for uncertainty in the system dynamics~\citep{XX-PT-JWG-ADA:15,XW-CB-CGC:22};
        \item CBF conditions that ensure safety under sample-and-hold controllers~\citep{JB-KG-DP:22};
        \item\label{it:cbf-conditions-slack} CBF conditions with added slack variables in order to ensure feasibility~\citep{JL-JK-ADA:23,ADA-SC-ME-GN-KS-PT:19}
    \end{enumerate}
    In all of these cases, the constraints can be expressed as affine inequalities in the control input, leading to a parametric QP and enabling the use of the closed-form solution discussed in Theorem~\ref{thm:piecewise-solution-parametric-qp}.
    \demo
\end{remark}

\smallskip


\begin{remark}
\label{rem:comparison-existing-literature}
    The closed-form expression provided in Theorem~\ref{thm:piecewise-solution-parametric-qp} coincides with those provided in the case of one or two constraints~\citep{XT-DVD:24,XX-PT-JWG-ADA:15}, or the one given in~\citep{MHC-EL-ADA:25} for box constraints on outputs with well-defined relative degree, under the choice of matrix $\bR$ detailed therein.
    \demo
\end{remark}

An additional benefit of the closed-form in Theorem~\ref{thm:piecewise-solution-parametric-qp} is that it enables detailed analysis of the regularity of $\bu^*$. 
\smallskip

\begin{corollary}
\label{prop:lipschitz-general}
Under Assumption~\ref{ass:assum1}:
\begin{enumerate}[(i)]
\item\label{it:on-region-smoothness} \textit{(On-region smoothness)} On each active–set interior $\text{Int}(\mathcal R_{\Ic})$,
$\bu^*(\bx)=\bu_{\Ic}(\bx)$ from \eqref{eq:lambda} and the map $\bx\mapsto \bu^*(\bx)$ is $C^1$.
\item\label{it:local-Lipschitzness} \textit{(Local Lipschitzness)} $\bu^*$ is locally Lipschitz on $\Xc$.
\end{enumerate}
\end{corollary}
\begin{proof}
To prove~(\ref{it:on-region-smoothness}), fix $\Ic$ and $\bx\in\text{Int}(\Rc_{\Ic})$. Since $\bx$ is in the interior of $\Rc_{\Ic}$, there exists a neighborhood $\Nc_{\bx}$ of $\bx$ where $\bu^*(\by) = \bu_{\Ic}(\by)$ for all $\by\in\Nc_{\bx}$.
Since $\bk, \{ L_{\bff}h_i, L_{\bg}h_i, \alpha(h_i) \}_{i\in[p]}$, are $C^1$, from~\eqref{eq:lambda} it follows that $\bu^*$ is $C^1$ at $\bx$.
On the other hand,~(\ref{it:local-Lipschitzness}) follows by~\citep{SMR:80} and because LICQ holds by
Assumption~\ref{ass:assum1}.
\end{proof}

Additional desirable regularity properties for CBF-QPs beyond Lipschitzness are discussed in~\citep{PM-AA-JC:25-ejc}. Having an explicit expression for~\eqref{eq:multiple-cbf-based-controller} paves the way for analyzing such properties. Here, it allows us to characterize the smoothness of $\bu^*$ within each region $\Rc_{\Ic}$ and to obtain sharper bounds on its Lipschitz constant. In particular, the Lipschitz constant can be computed directly for systems where the following assumptions hold.

\smallskip

\begin{assumption}\label{as:constant-affine-data}
    The functions $\{ \bb_i(\bx) \}$ are constant on $\Xc$, and $a_i, \bk$ are affine, i.e., 
    $\bb_i(\bx) \equiv \bb_i \in \real^m$,
    $a_i(\bx) = \bbgamma_i^\top \bx + \eta_i$, for some $\bbgamma_i\in\real^n$ and $\eta_i\in\real$,
    and $\bk(\bx) = \bK\bx + \bbkappa$ for all $\bx\in\Xc$.
\end{assumption}

Assumption~\ref{as:constant-affine-data} holds, for example, if the dynamics~\eqref{eq:control-affine-system} are linear and the constraints in~\eqref{eq:multiple-cbf-based-controller} arise from affine CBFs (or high-order CBFs~\citep{WX-CB:22}).
The next result shows that under Assumption~\ref{as:constant-affine-data}, the controller $\bu^*$ is piecewise affine and admits an explicit Lipschitz constant.
\smallskip

\begin{corollary}\longthmtitle{Piecewise affine structure}
\label{cor:piecewise-affine-structure}
    Under
Assumptions~\ref{ass:assum1} and~\ref{as:constant-affine-data}, for any $\Ic = \{ i_1, \hdots, i_k \} \subset [p]$ with $\bB_{\Ic} := [ \bb_{i_1}, \hdots, \bb_{i_k}]$ full column rank, define: 
    \begin{align*}
        \bH_{\Ic} := \bB_{\Ic}^\top \bR^{-1} \bB_{\Ic} \succ 0, \ \boldsymbol{\Gamma}_{\Ic}:={\small\begin{bmatrix} \bbgamma_{i_1}^\top\\ \vdots\\ \bbgamma_{i_k}^\top\end{bmatrix}}, \
        \boldsymbol{\eta}_{\Ic}:={\small\begin{bmatrix}\eta_{i_1}\\ \vdots\\ \eta_{i_k}\end{bmatrix}}.
    \end{align*}
    Then, the KKT quantities in~\eqref{eq:lambda} are affine in $\bx$: 
    \begin{align*}
    \blambda_{\Ic}(\bx) = \bG_{\Ic} \bx + \bg_{\Ic}, \quad 
        \bu_{\Ic}(\bx) = \bK_{\Ic} \bx + \bbkappa_{\Ic},
    \end{align*}
    where: 
    \begin{align*}
        &\bG_{\Ic} := \bH_{\Ic}^{-1}(\bB_{\Ic}^\top \bK + \boldsymbol{\Gamma}_{\Ic}), \quad \bg_{\Ic} = \bH_{\Ic}^{-1}(\bB_{\Ic}^\top \bbkappa + \boldsymbol{\eta}_{\Ic}), \\
        &\bK_{\Ic} = \bK - \bR^{-1} \bB_{\Ic} \bG_{\Ic}, \quad \bbkappa_{\Ic} = \bbkappa - \bR^{-1} \bB_{\Ic} \bg_{\Ic}.
    \end{align*}
    Hence, each region $\Rc_{\Ic}$ is a polyhedron:
    \begin{align}\label{eq:R_I-polyhedron}
        \notag
        &\Rc_{\Ic} = \setdef{\bx\in\Xc}{ \bG_{\Ic} \bx + \bg_{\Ic} \geq 0, \\
        &\qquad \quad  \bb_j^\top (\bK_{\Ic} \bx + \bbkappa_{\Ic}) + \bbgamma_j^\top \bx + \eta_j < 0, \forall j\notin \Ic },
    \end{align}
    and for all $\bx\in\Rc_{\Ic}$, we have $\bu^*(\bx) = \bK_{\Ic} \bx + \boldsymbol{\kappa}_{\Ic}$.
    Furthermore, for any convex subset $\bar{\Xc} \subset \Xc$, $\bu^*$ is Lipschitz in $\bar{\Xc}$ with a constant 
    $L:= \max\limits_{ \Ic : \Rc_{\Ic} \cap \bar{\Xc} \neq \emptyset } \norm{\bK_{\Ic}}$.
\end{corollary}
\begin{proof}
    The expression for $\bu^*$ follows directly from Theorem~\ref{thm:piecewise-solution-parametric-qp} and the expressions in Assumption~\ref{as:constant-affine-data}. To show that $\bu^*$ is $L$-Lipschitz in $\bar{\Xc}\subset\Xc$, take any $\bx,\by\in\bar{\Xc}$ and consider the line segment $\ell(s)=\bx+s(\by-\bx)$, $s\in[0,1]$.
    Note that since $\bar{\Xc}$ is convex,  $\ell(s)\in\bar{\Xc}$ for all $s\in[0,1]$.
    Let $\{ s_r \}_{r=0}^M$ (with $s_0 = 0$ and $s_1 = 1$) be the sequence of values of the parameter $s$ at which there is a region change (note that $M$ can either be finite or infinite).
    This means that for every $r\in\{ 0, \hdots, M-1 \}$, the segment $\{ \bx + s(\by-\bx) \}_{ s\in[ t_{r}, t_{r+1} ) }$ belongs to the same region (say $\Rc_{\Ic_r}$).
    Hence, $\bu^*(\ell(s)) = \bu_{\Ic_r}(\ell(s))$ for all $s \in [ s_{r}, s_{r+1} )$.
    Then, since $\bu^*$ is continuous at every $\{ \ell(s_r) \}_{r=0}^M$,
    $\|\bu^*(\ell(s_{r+1}))-\bu^*(\ell(s_r))\| = \bK_{\Ic_r} \norm{ \ell(s_{r+1})-\ell(s_r) }
    \le L \norm{ \ell(s_{r+1})-\ell(s_r) }$.
    Summing over $r=0,\dots,M-1$ and using the triangle inequality gives
    \begin{align*}
        \norm{\bu^*(\by)-\bu^*(\bx)}
        \le L\,\sum_{r=0}^M \norm{\ell(s_{r+1})-    \ell(s_r) }
        = L\,\norm{\by-\bx}.
    \end{align*}
    Hence $\bu^*$ is $L$-Lipschitz in $\bar\Xc$.
\end{proof}

Corollary~\ref{cor:piecewise-affine-structure} shows that under Assumption~\ref{as:constant-affine-data}, the controller $\bu^*$ is affine in each region $\Rc_{\Ic}$. In particular, this enables the use of linear control techniques to analyze the properties of the closed-loop dynamics (such as equilibria and their stability properties) within each region.
Additionally, Corollary~\ref{cor:piecewise-affine-structure} provides a tight Lipschitz constant for $\bu^*$ in any convex subset $\bar{\Xc}\subset\Xc$.
Knowledge of such Lipschitz constant is useful, for example, when implementing $\bu^*$ in a sampled-data fashion (cf. Section~\ref{sec:resource-aware-implementation}).

Another important feature of the controller $\bu^*$ under Assumption~\ref{as:constant-affine-data} is that the parameters $\bK_{\Ic}$ and $\bbkappa_{\Ic}$ defining the controller in each region $\Ic$ can be precomputed offline. 
Hence, the online implementation of $\bu^*$ in region $\Ic$ is straightforward. This is in contrast to the general case discussed in Theorem~\ref{thm:piecewise-solution-parametric-qp}, where the implementation of the controller $\bu^*$ requires the computation of the inverse of the matrix $\bB_{\Ic}(\bx)^\top \bR^{-1} \bB_{\Ic}(\bx)$ at every state $\bx\in\Rc_{\Ic}$.

\section{Resource-Aware CBF-QP Implementation}\label{sec:resource-aware-implementation}

Although Theorem~\ref{thm:piecewise-solution-parametric-qp} provides a closed-form expression for the controller $\bu^*$ in each region $\Rc_{\Ic}$, implementing this controller at a state $\bx\in\Xc$ requires knowing the active index set $\Ic^*(\bx)$.
Nevertheless, once $\Ic^*(\bx)$ is identified at a particular state, we may apply the closed-form controller $\bu_{\Ic}$ as long as the state remains within the region $\Rc_\Ic$. The results of Theorem~\ref{thm:piecewise-solution-parametric-qp} allow us to detect when the state trajectory exits this region. 
In this section, we formalize this idea and propose a practical, resource-aware implementation of the explicit expression of $\bu^*$.

Formally, we assume access to a computational method $\Theta:\Xc\to\Pc([p])$ that returns $\Ic^*(\bx)$ for any given state $\bx\in\Xc$.
This function $\Theta$ can be obtained through different methods such as:
\begin{itemize}
    \item Enumerating all active sets $\Ic \subset [p]$ and checking the conditions defining $\Rc_{\Ic}$ in~\eqref{eq:R-Ic-expression};
    \item Using a numerical solver to solve~\eqref{eq:multiple-cbf-based-controller} and extracting the indices of active constraints at the solution;
    \item Using a pre-trained classifier (e.g., a neural network).
\end{itemize}

Although the function $\Theta$ produces an exact prediction of the active set $\Ic^*(\bx)$ and enables the closed-form computation of the controller $\bu^*$ via Theorem~\ref{thm:piecewise-solution-parametric-qp},
evaluating $\Theta$ is often computationally expensive. Our goal is, therefore, to minimize the number of calls to $\Theta$ and obtain a resource-aware implementation of $\bu^*$.

We implement the controller in a sample-and-hold fashion. Let $\Delta T >0$ denote the sampling period. We compute the control input periodically at time instants $\{t_k\}_{k\in\mathbb{N}}$ with $t_k = k\Delta T$. For time $t\in[t_k,t_{k+1})$, the control input is held constant as $\bu(t) = \bu_k := \bu(t_k)$. In this setting, we develop an algorithm to reduce the number of sampling instants $t_k$ at which the function $\Theta$ is called.

\medskip
\begin{remark}{\rm
    In practice, the controller $\bu^*$ is implemented in a sample-and-hold fashion since digital platforms only permit discrete-time updates. The discrepancy of the trajectories between continuous-time and sample-and-hold implementation of  $\bu^*$ can be bounded using the Lipschitz constant of $\bu^*$ (which is now relatively easier to obtain using the closed-form of the QP, e.g., Corollary~\ref{cor:piecewise-affine-structure}) together with the sampling interval~$\Delta T$ via results such as ~\cite[Theorem 1]{AJT-VDD-RKC-YY-ADA:22}. The discrepancy becomes smaller as the Lipschitz constant and $\Delta T$ become smaller. We note also that although a sampled-data implementation of $\bu^*$ might violate the safety constraints during inter-event times, the CBF constraints in~\eqref{eq:multiple-cbf-based-controller} can be modified as in~\citep{JB-KG-DP:22,AS-YC-ADA:20,AJT-VDD-RKC-YY-ADA:22} to remedy that and ensure safety during inter-event times.} 
    \demo
\end{remark}

Algorithm~\ref{alg:resource-aware-explicit-QP} formalizes our proposed procedure for computing the controller $\bu_k$ at each sampling instant $t_k$:
\begin{algorithm}[t]
\caption{Resource-Aware Explicit CBF-QP}
\label{alg:resource-aware-explicit-QP}
\begin{algorithmic}[1]
\REQUIRE $\bx$, $\Ic$, $\bk$, $\bR$, $\bff$, $\bg$, $\{ h_i \}_{i=1}^p$.
\IF {$\rank(\bB_{\Ic}(\bx)) = |\Ic|$ and $S_{\Ic}^{(1)}(\bx) \geq 0$ and $S_{\Ic}^{(2)}(\bx) < 0$}
    \RETURN $\bu_{\Ic}(\bx)$, $\Ic$
\ELSE
    \STATE Let $\Ic^{\prime} = \Theta(\bx)$
    \RETURN $\bu_{\Ic^{\prime}}(\bx)$, $\Ic^{\prime}$
\ENDIF
\end{algorithmic}
\end{algorithm}

\emph{Lines 1 - 2:} Given a state $\bx$ and an index set $\Ic$, Line 1 verifies whether $\rank(\bB_{\Ic}(\bx)) = |\Ic|$ and the conditions $S_{\Ic}^{(1)}(\bar{\bx}) \geq 0$ and $S_{\Ic}^{(2)}(\bx) < 0$ are verified, where 
\begin{align}\label{eq:trigger-definition}
    S_\Ic^{(1)}(\bx) = \blambda_{\Ic}(\bx), \ S_\Ic^{(2)}(\bx) = \max\limits_{i\notin \Ic} \{ \bb_i(\bx)^\top \bu_{\Ic}(\bx) + a_i(\bx) \}.
\end{align}
If these conditions are satisfied, then by Theorem~\ref{thm:piecewise-solution-parametric-qp} item~\ref{it:active-set-region}, $\bx\in\Rc_{\Ic}$, and Line 2 returns the control input $\bu_{\Ic}(\bx)$ and the active set $\Ic$.

\emph{Lines 3 - 5:} If either $\rank(B_{\Ic}(\bx)) \neq |\Ic|$, $S_{\Ic}^{(1)}(\bx) < 0$ or $S_{\Ic}^{(2)}(\bx) \geq 0$, then $\bx\notin\Rc_{\Ic}$. In this case, we use the function $\Theta$ to obtain the active set at $\bx$, denoted $\Ic^{\prime}$ (Line 4), and return it along with the corresponding control input $\bu_{\Ic^{\prime}}(\bx)$ (cf. Line 5).

The following result shows that Algorithm~\ref{alg:resource-aware-explicit-QP} produces the correct controller at each sampling time.

\smallskip

\begin{proposition}\longthmtitle{Correctness of Resource-Aware Explicit CBF-QP}
\label{prop:correctness-resource-aware-explicit-cbf-qp}
    Consider the parametric QP \eqref{eq:multiple-cbf-based-controller} under Assumption~\ref{ass:assum1}.
    Let $\bx\in\Xc$ and $\Ic\subset[p]$.
    Then, Algorithm~\ref{alg:resource-aware-explicit-QP} outputs $(\bu^*(\bx), \Ic^*(\bx))$.
\end{proposition}
\begin{proof}
    First, suppose that $\Ic = \Ic^*(\bx)$.
    In this case, by the characterization of $\Rc_{\Ic}$ in Theorem~\ref{thm:piecewise-solution-parametric-qp}, the if statement in Line 1 is true. Therefore, the algorithm returns $(\bu_{\Ic}(\bx), \Ic) = (\bu^*(\bx), \Ic^*(\bx))$ (cf. Line 2).
    Alternatively, suppose that $\Ic \neq \Ic^*(\bx)$.
    Then, by the the characterization of $\Rc_{\Ic}$ in Theorem~\ref{thm:piecewise-solution-parametric-qp}, the if statement in Line 1 is false.
    Since $\Ic^{\prime} = \Theta(\bx) = \Ic^*(\bx)$, the algorithm returns $(\bu_{\Ic^{\prime}}(\bx), \Ic^{\prime}) = (\bu^*(\bx), \Ic^*(\bx))$ (cf. Line 5).
\end{proof}

Proposition~\ref{prop:correctness-resource-aware-explicit-cbf-qp} shows that by executing Algorithm~\ref{alg:resource-aware-explicit-QP}, we may be able to bypass invoking $\Theta$ (e.g., solving the QP numerically) at every sampling time $t_k$ while still obtaining the controller $\bu^*(\bx(t_k))$ and the corresponding active set $\Ic^*(\bx(t_k))$. Instead, at each step, we may simply guess the active set $\Ic$ and run the algorithm with this guess. Verifying whether the guess is correct requires checking only for a single candidate (rather than all $\Ic\in[p]$), and  is therefore computationally inexpensive. 

The approach is most effective when the algorithm is called with a correct guess of $\Ic$. To encourage this, we leverage the fact that system trajectories evolve continuously and typically remain within a region $\Rc_\Ic$ for some time. Thus, feeding the previously identified active set $\Ic^\prime$ back as an input to the algorithm is a natural and often accurate choice. In practice, cf. Section~\ref{sec:simulations}, this leads to an order-of-magnitude reduction in the number of times $\Theta$ is called. Future work includes more advanced guessing strategies, e.g., via machine learning, to further improve efficiency.

\section{Simulations}\label{sec:simulations}

In this section, we present  simulations illustrating the applicability of the closed-form introduced in Theorem~\ref{thm:piecewise-solution-parametric-qp} and the resource-aware implementation in Algorithm~\ref{alg:resource-aware-explicit-QP}.

\subsection{Comparison of execution times}\label{sec:comparison-exec-times}

As mentioned in Section~\ref{sec:resource-aware-implementation}, the main difficulty of implementing the closed-form in Theorem~\ref{thm:piecewise-solution-parametric-qp} at an arbitrary state $\bx\in\Xc$ is that it requires knowing the active set $\Ic^*(\bx)$ at $\bx$ (or equivalently, the region $\Rc_{\Ic^*(\bx)}$ that $\bx$ belongs to).
Here we compare the execution times of two alternative methods of implementing $\bu^*$ at an arbitrary $\bx\in\Xc$:

\begin{enumerate}
    \item\label{it:alg-solve-qp} \textit{Solver:}  solving the QP with an optimization algorithm (we use the state-of-the-art OSQP algorithm~\citep{BS-GB-PG-AB-SB:20});
    \item\label{it:alg-enumerate-all} \textit{Explicit:} enumerating all possible active sets $\Ic\subseteq[p]$, checking whether $\bx\in\Rc_{\Ic}$ (through~\eqref{eq:R-Ic-expression}), and then executing the corresponding closed-form controller.
\end{enumerate}


\begin{figure}
    \centering
    \includegraphics[width=1.0\linewidth]{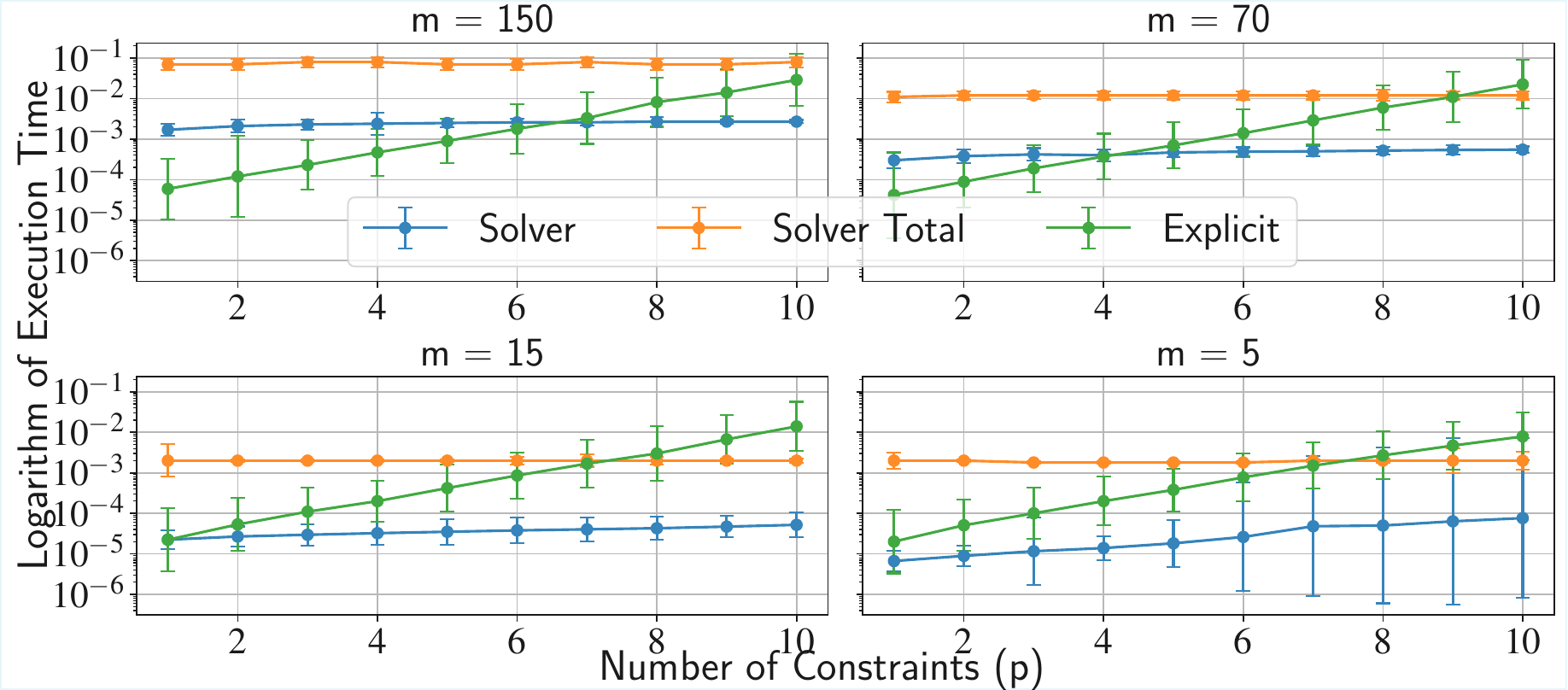}
    \caption{Comparison of methods~\ref{it:alg-solve-qp} and~\ref{it:alg-enumerate-all} for computing $\bu^*$.
    \textit{Solver} refers to the execution time of the OSQP algorithm, whereas \textit{Solver Total} refers to the total elapsed time in a call to the OSQP solver, including problem setup, solver execution, and post-processing.}
    \label{fig:comparison_execution_times}
\end{figure}

For different values of the dimension of the optimization variable ($m$) and number of constraints ($p$), we generate 500 different QPs with randomly generated coefficients.
Figure~\ref{fig:comparison_execution_times} shows the average execution time (with one–standard deviation bars) for methods~\ref{it:alg-solve-qp} and~\ref{it:alg-enumerate-all}. In the figure, we have included execution time of \textit{Solver} both with and without the solver setup time since
both measures need to be considered in applications.

We observe that the \textit{Explicit} method is faster when the dimension of the optimization variable is large and the number of constraints is small.
Hence, Figure~\ref{fig:comparison_execution_times} illustrates the regimes in which each of the different methods should be used as our computation method~$\Theta$ in Algorithm~\ref{alg:resource-aware-explicit-QP}: \textit{Explicit} for high-dimensional optimization variables and small number of constraints, and \textit{Solver} for low-dimensional optimization variables with many constraints.

\subsection{Servo-control of the roll-yaw dynamics of an aircraft}\label{sec:servo-control-lateral-dynamics-aircraft}

We consider the roll-yaw dynamics of a mid-size aircraft around an operating point (cf.~\citep[Section 14.8]{EL-KAW:24}) with velocity $717.17$ ft/sec, altitude $25000$ ft, and angle of attack $4.5627^\circ$.
The state is $\bx_p=[\beta, p_s, r_s]^\top$, where $\beta$ is sideslip (rad) and $p_s$, $r_s$ are roll and yaw rates (rad/s). The inputs are aileron and rudder deflections $\delta_a$ and $\delta_r$ (rad). The dynamics are
\begin{align*}
    \dot{\bx}_p = \bA_p \bx_p + \bB_p \bu,
\end{align*}
with
\begin{align*}
    \bA_p \! = \! {\small\begin{bmatrix}
        -0.1179 & 0.0009 & -1.001 \\
        -7.0113 & -1.4492 & 0.2206 \\
        6.3035 & 0.0651 & -0.4117
    \end{bmatrix}}, \ 
    \bB_p \! = \! {\small
\begin{bmatrix}
    0 & 0.0153 \\
    -7.9662 & 2.6875 \\
    0.6093 & -2.3577
\end{bmatrix}.
}
\end{align*}

The outputs are the roll rate $p_s$ (rad/s) and lateral load factor $N_y$ ($g$), with $g = 32.174$ ft/s$^2$. They are given by
\begin{align*}
    \by_{\text{reg}} = \begin{bmatrix}
        p_s \\
        N_y
    \end{bmatrix} = 
    \bC_{\text{p,reg}} \bx_p + \bD_{\text{p,reg}} \bu,
\end{align*}
with 
\begin{align*}
    \bC_{\text{p,reg}} = {\small\begin{bmatrix}
        0 & 1 & 0 \\
        -2.6049 & 0.0187 & 0.0677
    \end{bmatrix}}, \ 
    \bD_{\text{p,reg}} = {\small\begin{bmatrix}
        0 & 0 \\
        0 & 0.3370
    \end{bmatrix}}.
\end{align*}
Our goal is to regulate these outputs to a commanded signal $\by_{\text{cmd}}:\real_{\geq 0} \to \real^2$.
We design a baseline LQR PI controller using the integrated output tracking error dynamics (cf.~\citep[Section 4.4.1]{EL-KAW:24}), with $\bQ = \text{diag}(1.025, 1.029, 0.0, 0.0, 1.602)$ and $\bR = \text{diag}(1, 1)$.


Next, we set roll rate limits to $\pm 0.4$, and yaw rate limits to $\pm 0.05$.
Additionally, in order to prevent the integrator states from winding up (which would lead to poor tracking performance), we also constrain the integrator errors to be in the interval $\pm 0.3$.
In order to do so, we follow an approach similar to that of~\citep{EL:25} and introduce a virtual control input $\bv$ for the error variables as follows:
\begin{align*}
    \dot{\be}_{yI} = \by_{\text{reg}} - \by_{\text{cmd}} + \bv.
\end{align*}
The augmented system in the $[\be_{yI}, \bx_p]$ variables is:
{\small\begin{align*}
    \begin{bmatrix}
        \dot{\be}_{yI} \\
        \dot{\bx}_p
    \end{bmatrix} = 
    \begin{bmatrix}
        \mathbf{0}_{m\times m} & \bC_{\text{p,reg}} \\
        \mathbf{0}_{n_p \times m} & \bA_p
    \end{bmatrix} 
    \begin{bmatrix}
        \be_{yI} \\
        \bx_p
    \end{bmatrix}
    + 
\begin{bmatrix}
        \bD_{\text{p,reg}} \\
        \bB_p
    \end{bmatrix} \bu
    + 
\begin{bmatrix}
        -\mathbf{I}_m \\
        0
    \end{bmatrix} (\by_{\text{cmd}} - \bv).
\end{align*}}

As noted earlier, the use of QP solvers is undesirable in aerospace applications, so we use Algorithm~\ref{alg:resource-aware-explicit-QP}.
Here,  $\Theta$ is computed by checking all active sets of cardinality less than $2$ (which is the number of inputs and hence the maximum possible number of active constraints at the optimizer if LICQ holds).
Figure~\ref{fig:filtered-output-evolution} 
showcases the evolution of the output variables, which as expected remain within the operational limits.
In order to guarantee feasibility of the QP, we introduce slack variables and penalize them in the objective function, as described in Remark~\ref{rem:general-parametric-qps}, item~\ref{it:cbf-conditions-slack}. In this example, the function $\Theta$ is only called 8 times (out of the total of 3000 sampling times).

\begin{figure}
    \centering
    \includegraphics[width=0.99\linewidth]{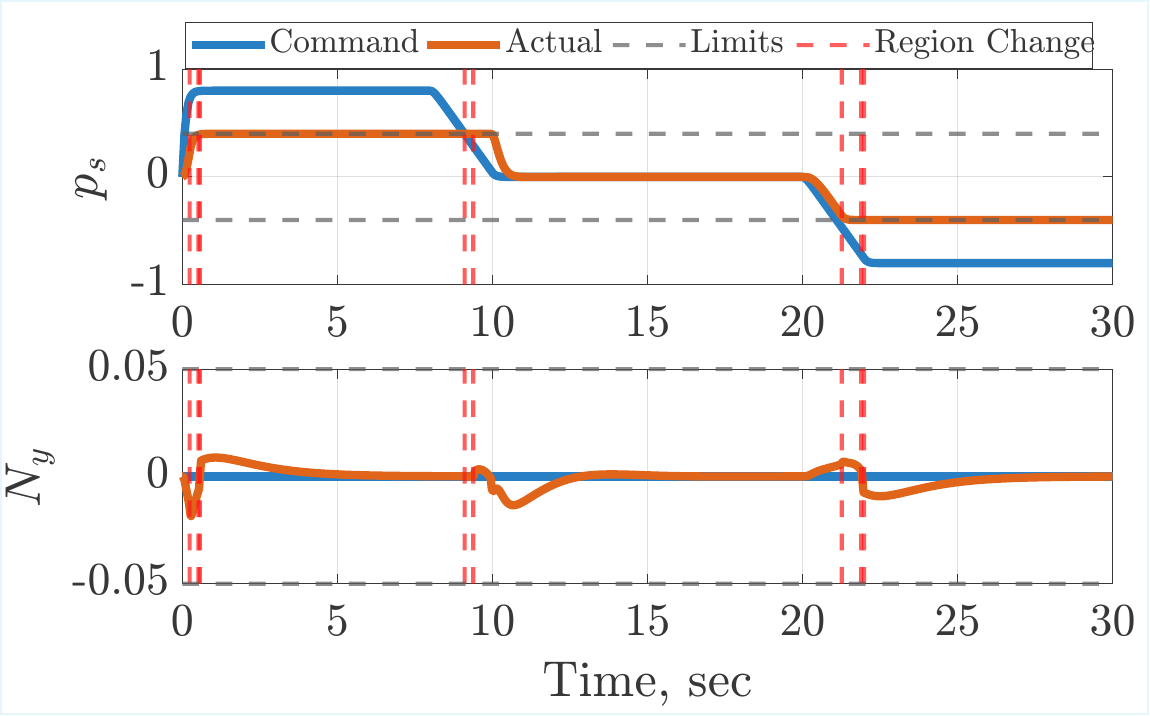}
    \caption{Output evolution under the CBF-based safety filter for the example in Section~\ref{sec:servo-control-lateral-dynamics-aircraft}.}
    \label{fig:filtered-output-evolution}
\end{figure}


\subsection{Navigation of a team of quadrotors}\label{sec:navigation-team-quadrotors}

We consider three quadrotors navigating in a shared environment with obstacles. We model each quadrotor using a cascaded control architecture separating translational and rotational dynamics. A 12-dimensional full-order quadrotor model is given by 
\begin{align*}
\underbrace{\left[\begin{array}{c}
\dot{\mathbf{p}}  \\
\dot{\mathbf{v}} \\
\dot{\mathbf{\eta}} \\
\dot{\mathbf{\omega}}
\end{array}\right]}_{\dot{\mathbf{x}}} = \underbrace{\left[\begin{array}{c}
\mathbf{v} \\
-\mathbf{e}_{3} g \\
\mathbf{S}(\mathbf{\eta}) \mathbf{\omega} \\
-\mathbf{J}^{-1} (\mathbf{\omega} \times (\mathbf{J} \mathbf{\omega}))
\end{array}\right]}_{\mb{f}(\mb{x})} 
+ 
\underbrace{\begin{bmatrix}
\mathbf{0}_{3 \times 1} &  \mathbf{0}_{3 \times 3} \\
   \frac{1}{m} \tilde{\mathbf{R}}(\mathbf{\eta}) \mathbf{e}_{3} &\mathbf{0}_{3 \times 3} \\
   \mathbf{0}_{3 \times 1} &  \mathbf{0}_{3 \times 3} \\
\mathbf{0}_{3 \times 1} & \mathbf{J}^{-1}
\end{bmatrix}}_{\mb{g}(\mb{x})}
\underbrace{\begin{bmatrix}
   {F} \\ \mathbf{M}  
\end{bmatrix}}_{\mb{u}},
\end{align*}
where $\mathbf{p} \in \mathbb{R}^{3}$ 
is the position in an inertial frame, $\mathbf{v} \in \mathbb{R}^{3}$ the linear velocity, $\mathbf{\eta} \in \mathbb{S}^1 \times \mathbb{S}^1 \times \mathbb{S}^1$  the roll-pitch-yaw angles, and $\mathbf{\omega} \in \mathbb{R}^{3}$ the body angular velocity. The constants $g$ and $m$ are the gravity and the quadrotor mass, and $\mathbf{J} \in \mathbb{R}^{3 \times 3}$ is the inertia matrix. The system has inputs of body moment vector $\mathbf{M} \in \mathbb{R}^{3}$ and thrust force $F \in \mathbb{R}$. Here $\mathbf{e}_{3}$ is the unit vector in the inertial $z$ direction, $\tilde{\mathbf{R}}: \mathbb{S}^1 \times \mathbb{S}^1 \times \mathbb{S}^1 \to \mathrm{SO(3)} $ is the rotation mapping, and $\mathbf{S}: \mathbb{S}^1 \times \mathbb{S}^1 \times \mathbb{S}^1 \to \mathbb{R}^{3 \times 3}$ is the kinematic mapping relating Euler angular rates and body angular velocity.

For each quadrotor $i \in [3] $, we assume that desired position, velocity, and acceleration ${\mathbf{p}_{r,i}}(t), {\dot{\mathbf{p}}_{r,i}(t)}, {\ddot{\mathbf{p}}_{r,i}(t)}$ are given, and $\mathbf{p}_{r,i}$ is continuous and at least twice differentiable. Given a reference trajectory, we utilize a virtual control input for position control:
\begin{align*}
    \mathbf{a}_{des,i} = \mathbf{K}_{p,i}^{\mathbf{p}}({\mathbf{p}_{r,i}} - \mathbf{p}_{i}) +  \mathbf{K}_{d,i}^{\mathbf{p}}({\mathbf{v}_{r,i}} - \mathbf{v}_{i}) + \ddot{\mathbf{p}}_{r,i},
\end{align*}
where $\mathbf{K}_{p,i}^{\mathbf{p}}, \mathbf{K}_{d,i}^{\mathbf{p}} \in \mathbb{R}^{3 \times 3}$ are diagonal gains. Writing $\mathbf{a}_{des,i} = [{a}_{x,i}~{a}_{y,i}~{a}_{z,i}]^\top$, 
the virtual nominal control input is 
$\bk_i(\bx,t) = m[a_{x,i}, a_{y,i}, g+a_{z,i}]$.
We let the concatenation of $\{ \bk_{i}(\bx,t) \}_{i=1}^3$ be $\bk(\bx,t)\in\real^9$, which is the desired control input that would need to be applied if the dynamics of the three quadrotors were a double integrator.

We enforce collision-avoidance safety by maintaining distance between the quadrotors and spherical obstacles at the position control layer via exponential CBFs (ECBFs) \citep{QN-KS:16}. For every pair of quadrotor $i$ and obstacle $j \in [N_o]$ centered at $\mathbf{c}_j$, we define a CBF: $h_{o,ij}(\mb{x}) = \| \mathbf{p}_{i} - \mathbf{c}_j\|_2^2 - d_{i,j}^2$, where $d_{i,j} \in \real_{\geq 0}$ is the safe distance. Similarly, for each pair of distinct quadrotors $(i,k)$ we define $h_{q,ik}(\mb{x}) = \| \mathbf{p}_{i} - \mathbf{p}_{k}\|_2^2 - d_{i,k}^2$, where $d_{i,k} \in \real_{\geq 0}$ is chosen from the quadrotor radius. Each CBF has relative degree two with respect to the virtual nominal control input $\bk$. Therefore, we obtain an ECBF-QP with safety constraints and input bounds. 
This yields a controller $\bu^*$ following~\eqref{eq:multiple-cbf-based-controller}, which we then actuate with $F$, $\mathbf{M}$ in the full quadrotor dynamics through an inner attitude control layer, following the methods presented in~\citep{RM-VK-PC:12}.

In our simulation~\footnote{\label{footnote:code}The code and additional details of our simulation can be found at \url{https://github.com/ersindas/Quadrotor_CBF}}, 
the workspace of three quadrotors contains $N_o = 16$ obstacles. Thus, the safety filter  has a total $3(N_o + 2) = 54$ ECBF inequalities and 18 input constraints. We choose the reference trajectories to generate frequent encounters. The first quadrotor tracks an 8-shaped path, while the others follow circular orbits in opposite directions. The results for this scenario are shown in Figure~\ref{fig:quadrotors}.   
Figure~\ref{fig:quadrotors_res} shows that using the resource-aware implementation in Algorithm~\ref{alg:resource-aware-explicit-QP} results in a controller execution time that is $6.74$ times faster on average compared to using a QP solver.
In this example, the function $\Theta$ is obtained by using the QP solver and only needs to be used in 112 out of the 2000 sampling times.

\begin{figure}[t]
    \centering
    \includegraphics[width=.8\linewidth]{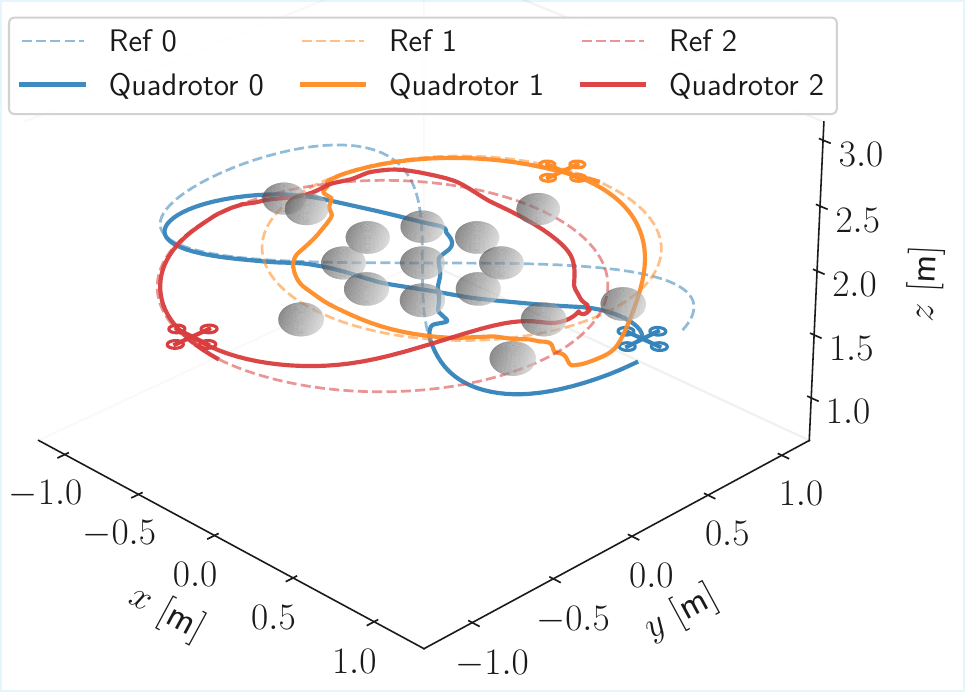}
    \caption{Reference and filtered trajectories defined by the team of quadrotors.}
    \label{fig:quadrotors}
\end{figure}

\begin{figure}[t]
    \centering
    \includegraphics[width=1\linewidth]{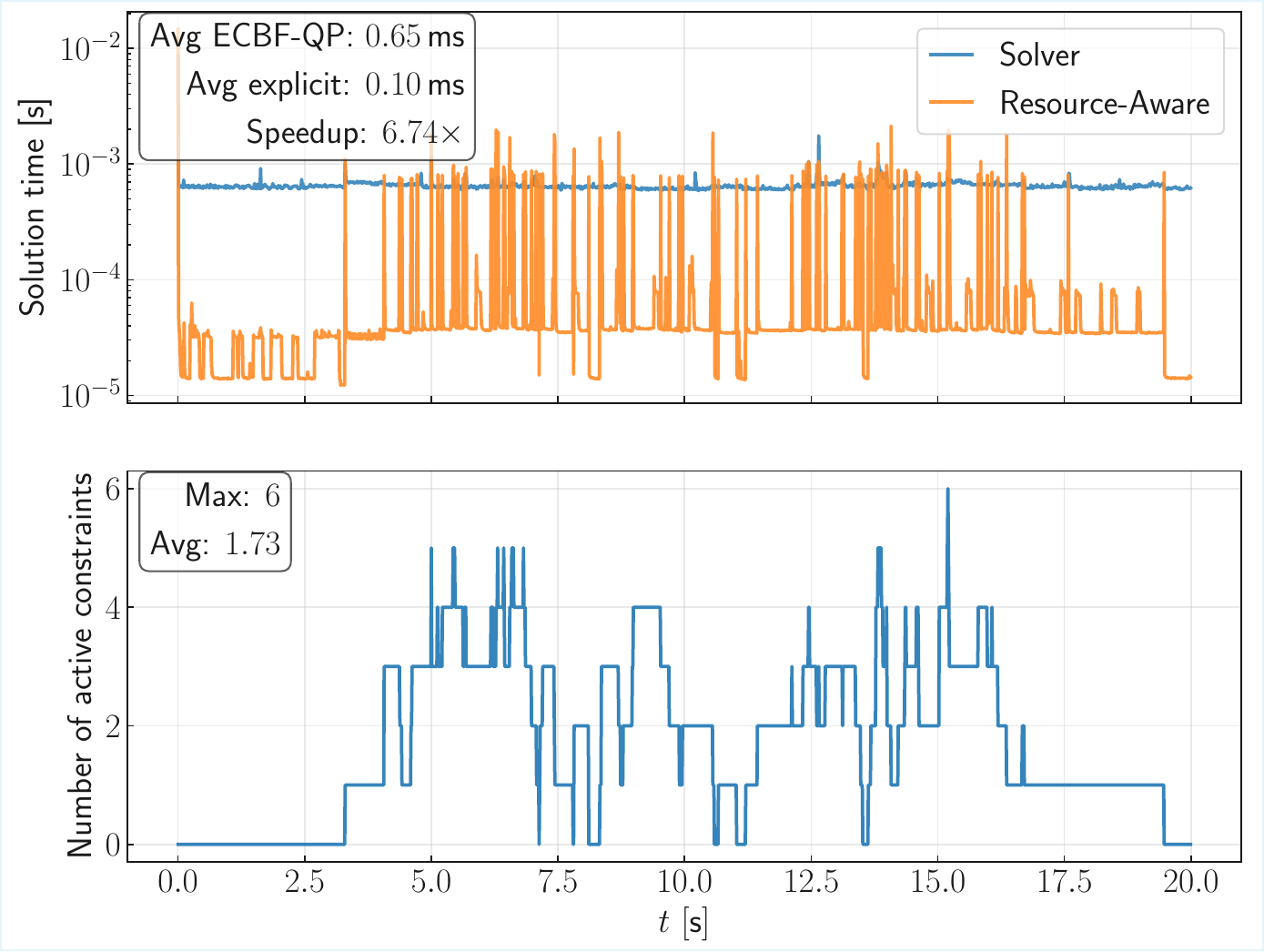}
    \caption{(Top) comparison of the execution time between \textit{Solver} (using the OSQP algorithm) and the \textit{Resource-Aware} implementation in Algorithm~\ref{alg:resource-aware-explicit-QP}. (Bottom) Evolution of the number of active indices.}
    \label{fig:quadrotors_res}
\end{figure}

\subsection{Safety filters in reinforcement learning training}\label{sec:safety-filtering-during-training-rl}

We consider an agent in a 2D environment following single integrator dynamics $\dot{\bx} = \bu$, with $\bx,\bu\in\real^2$.
The environment is a square region with different circular obstacles.
The goal of the agent is to get to a specific location in the environment (cf. Figure~\ref{fig:rl-environmen}) and avoid collisions with the walls and obstacles at all times.


\begin{figure}
    \centering
    \includegraphics[width=0.48\linewidth]{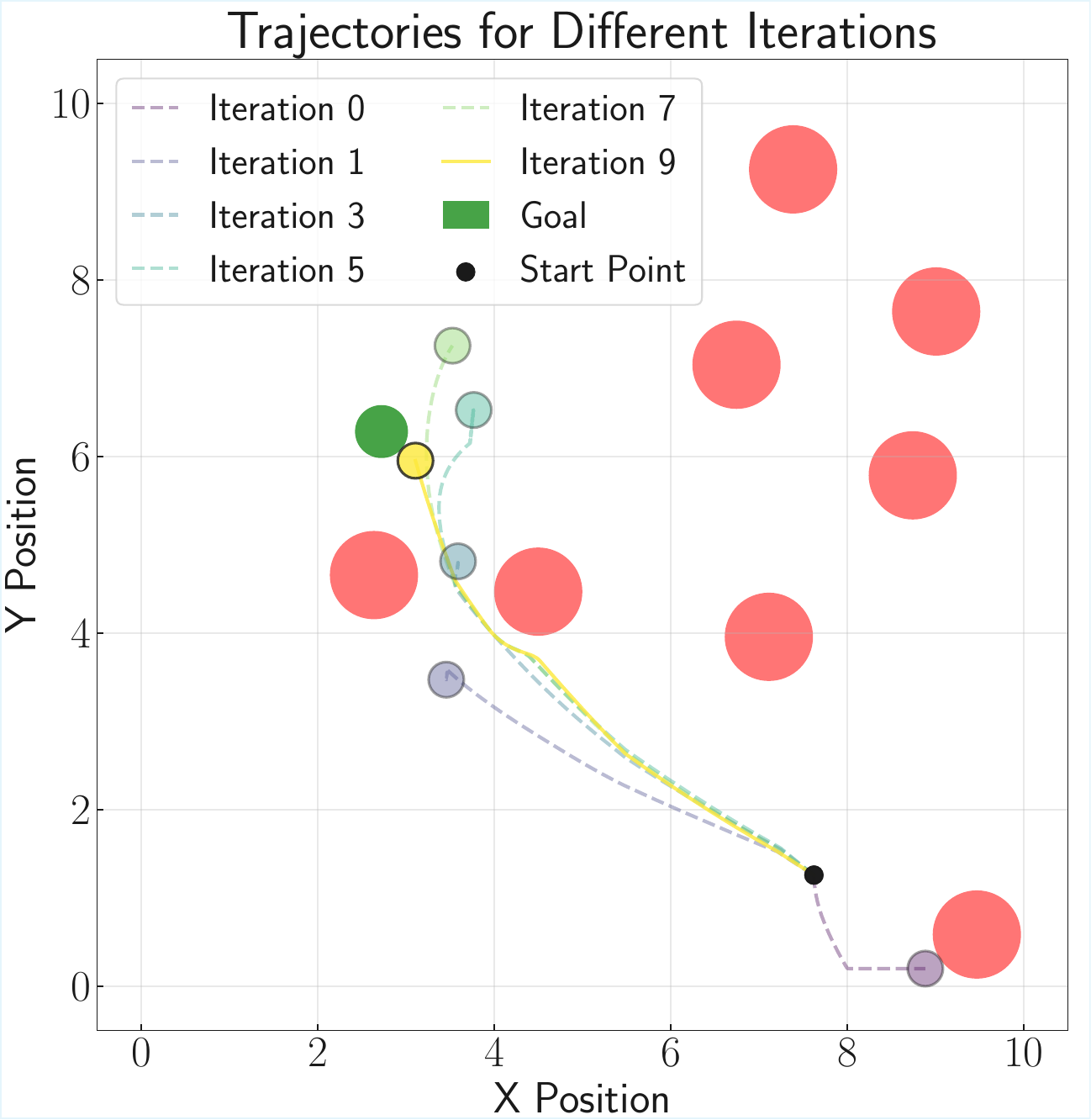}
    \includegraphics[width=0.48\linewidth]{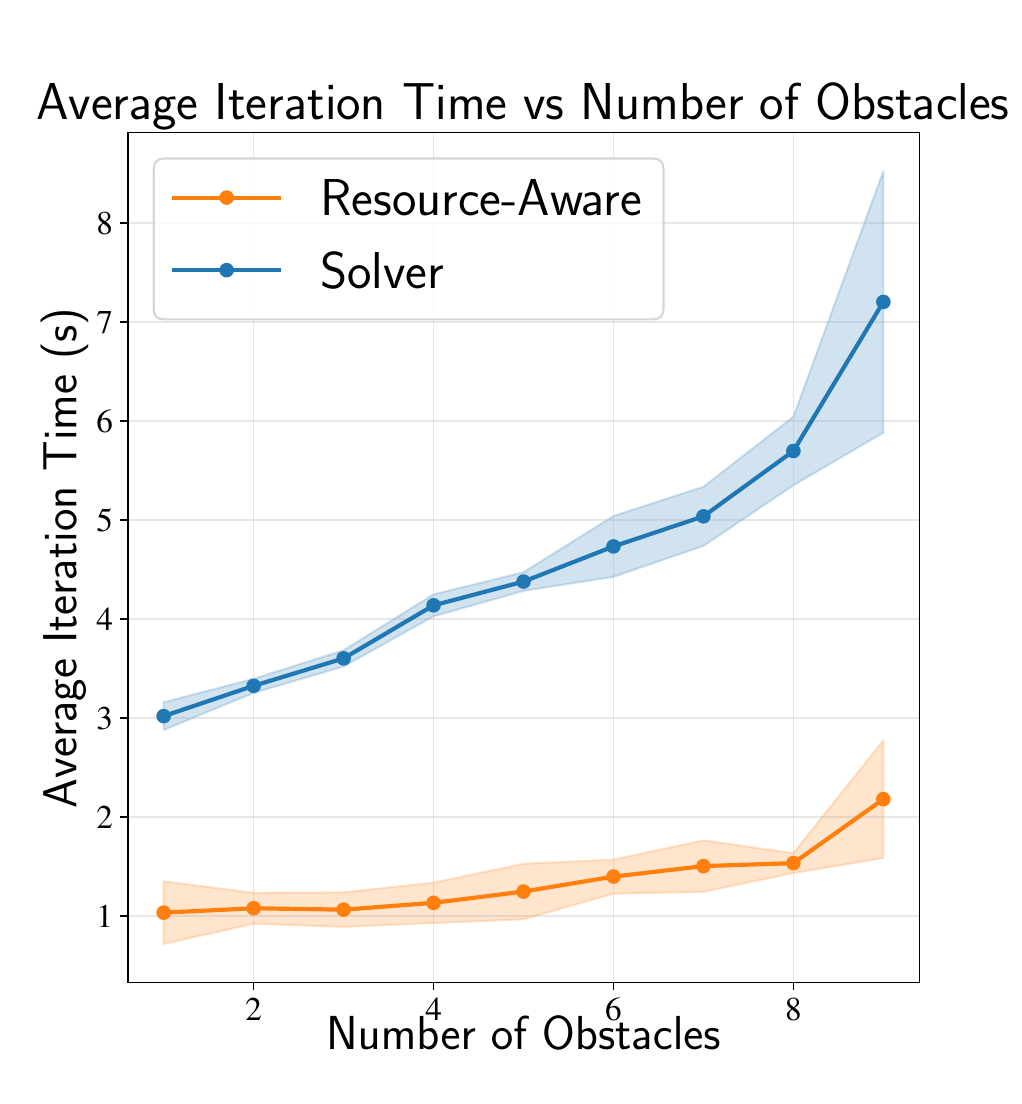}
    \caption{(left) Example of the type of environment used during the RL training process. The trajectories obtained from policies at different iterations are plotted in dashed lines. (right) Average iteration time of using QP solvers during training versus using Algorithm~\ref{alg:resource-aware-explicit-QP}. The shaded area represents the error bars, computed over 10 different trials.}
    \label{fig:rl-environmen}
\end{figure}



To solve this task, we use the CBF-RL algorithm~\citep{LY-BW-MdS-ADA:25}, which trains a nominal RL policy and safety filters its actions using CBFs defining the walls and obstacles in the environment as in~\eqref{eq:multiple-cbf-based-controller}.
This ensures that the RL agent remains safe throughout the training process.
The RL agent is trained over multiple environments similar to those in Figure~\ref{fig:rl-environmen} (left). The nominal RL policy is trained with Proximal Policy Optimization (PPO)~\citep{JS-FW-PD-AR-OK:17}  using 4096 parallel environments for 10 steps.
We compare the use of a QP solver and the resource-aware implementation of the explicit CBF-QP in Algorithm~\ref{alg:resource-aware-explicit-QP} during the RL training process.


Figure~\ref{fig:rl-environmen} (right) shows that computing~\eqref{eq:multiple-cbf-based-controller} using Algorithm~\ref{alg:resource-aware-explicit-QP} results in a considerable improvement in the average iteration time during the training process.
Furthermore, the improvement increases as the number of obstacles in the environment increases.
Figure~\ref{fig:rl-environmen} (right) shows that Algorithm~\ref{alg:resource-aware-explicit-QP} enables the use of CBF-RL in environments with multiple CBF constraints (the fact that using QP solvers during training is impracticable for CBF-RL was already recognized in~\citep{LY-BW-MdS-ADA:25}, which is why the safety filters therein are always limited to a single CBF constraint).

\section{Conclusions}
In this paper we have introduced a closed-form expression for the controller obtained from CBF-based safety filters. This formula is derived by partitioning the state-space into different regions, depending on the set of active constraints at the optimizer of the safety filter, and using a different closed-form expression in each region.
By leveraging this formula, we have introduced a resource-aware implementation of CBF-based safety filters, which bypasses the use of QP solvers at every state by detecting region changes and using the closed-form expression within each region.
We have illustrated the benefits of this novel implementation in a variety of examples, including aerospace control, navigation of a team of quadrotors, and safe RL.
Future work will focus on improving the efficiency of the resource-aware implementation by using pre-trained classifiers for the function $\Theta$ and applying the proposed framework in other CBF techniques, such as backup or robust CBFs.

\bibliography{bib/alias_brief,bib/Main-add,bib/Main,bib/JC,bib/New}

\end{document}